\documentclass[journal,10pt]{IEEEtran}

\usepackage{fancyhdr}
\usepackage{hyperref}

\usepackage[text={18cm,24.8cm}]{geometry}
\usepackage{graphicx}
\usepackage{multicol}
\usepackage{amssymb}
\usepackage{color}
\usepackage{amsmath}
\usepackage{amsfonts}
\usepackage{amsthm}
\usepackage{epstopdf}
\usepackage{subfig}
\usepackage[numbers,square,sort&compress]{natbib}
\newtheorem{thm}{Theorem}

\newtheorem{prpn}[thm]{Proposition}
\def\wh{\widehat}
\def\t{\tilde}
\def\mb{\mathbf}
\def\bP{\mathbf{P}}
\def\cP{\mathcal{P}}

\begin{document}

\title{Robust Power Allocation and Outage Analysis for Secrecy in Independent Parallel Gaussian Channels}
\author{\IEEEauthorblockN{Siddhartha Sarma, Kundan Kandhway and Joy Kuri}
\thanks{The authors are with the Department of Electronic Systems Engineering, Indian Institute of Science, Bangalore, Karnataka - 560012, India (e-mail: \{siddharth, kundan, kuri\}@dese.iisc.ernet.in).} \vspace*{-3ex}}

\maketitle

\thispagestyle{fancy}

\begin{abstract} 
This letter studies parallel independent Gaussian channels with uncertain eavesdropper channel state information (CSI). Firstly, we evaluate the probability of zero secrecy rate in this system for (i) given instantaneous channel conditions and (ii) a Rayleigh fading scenario. Secondly, when non-zero secrecy is achievable in the low SNR regime, we aim to solve a robust power allocation problem which minimizes the outage probability at a target secrecy rate. We bound the outage probability and obtain a linear fractional program that takes into account the uncertainty in eavesdropper CSI while allocating power on the parallel channels. Problem structure is exploited to solve this optimization problem efficiently. We find the proposed scheme effective for uncertain eavesdropper CSI in comparison with conventional power allocation schemes.
\end{abstract}
\vspace*{-2ex}
\section{Introduction}
By using inherent random noise in communication channels, physical layer security achieves \emph{information theoretically secure} communications. Researchers have studied and characterized \textit{secrecy capacity} for different communication systems and channel scenarios ranging from single antenna single hop \cite{barrosISIT2006} to multi-antenna multi-hop systems \cite{ngTWC2011}. Later, to improve secrecy capacity, researchers have proposed schemes like MIMO with artificial noise generation \cite{goelTWC2008}, jammer assisted transmission \cite{huangTSP2012}, cooperative relaying \cite{dong1}, analog network coding \cite{chenTIFS2012} and combined relaying-jamming \cite{parkJSAC2013} in the context physical layer security.

However, most of the existing literature on physically secure communications considers perfect knowledge of eavesdroppers' channel state information (CSI)---a far fetched assumption. For real world scenarios, e.g., border surveillance, we can only expect partial eavesdropper CSI (e.g. estimated path loss). Recently a few papers have discussed power allocation to improve secrecy for single channel scenarios when no CSI or partial CSI for the eavesdropper's channel is available, either with the help of a jammer or using beamforming or both \cite{ngTWC2011,huangTSP2012,wang2013joint}. But studies involving (robust) optimal power allocation for \emph{parallel Gaussian channels with imperfect eavesdropper CSI} have received little attention. 

The parallel channels serve as a model for wideband wireless communications, channels with inter-symbol interference, block fading channels and multi-antenna systems. The secrecy capacity of parallel channels was studied in \cite{zang2010} and optimal power allocation for the Gaussian scenario was evaluated in \cite{jor}. But none of them addressed the imperfect eavesdropper CSI scenario. In the current article, we propose a robust power allocation scheme which ensures minimum secrecy outage when partial eavesdropper CSI is available. 

Robust power allocation has appeared in \cite{mallickTOC2013} for relay channels without secrecy, in \cite{huangTSP2012} for MISO systems, and in \cite{liICASSP2014} for amplify and forward relaying in the context of secrecy. However, these works did not consider parallel Gaussian channels.
Our contributions are summarized below.
\begin{itemize}
\item Approximate \textit{instantaneous complete secrecy outage} probability for partial eavesdropper CSI. \textit{Closed form} expression for average complete secrecy outage for fading channels.
\item When non-zero secrecy is possible, optimal power allocation to minimize $\Pr(R_s<R_s^{(0)})$, where $R_s^{(0)}$ is the target secrecy rate. The proposed technique for this \emph{robust power allocation problem} bounds the outage probability and leads to a linear fractional program.
\item \emph{Computationally efficient technique} to solve the formulated linear fractional program by exploiting the problem structure. Comparison of this power allocation technique with several conventional schemes with respect to secrecy outage.  
\end{itemize}

\vspace*{-3ex}
\section{System Model}
\label{sec:Model}
We consider a single transmitter (source)--receiver (destination) pair in presence of an eavesdropper. The source can transmit information to the destination using $N$ parallel channels indexed by $i \in \mathcal{N}=\{1,2,\cdots,N\}$. The eavesdropper is passively listening to the source--destination transmission. The $i$th channel gains for the source to destination channel and the source to eavesdropper channel are denoted by complex numbers $h_i$ and $g_i$, respectively. The incomplete CSI for the eavesdropper's channel is modeled as: $g_i=\wh{g}_i+\t{g}_i$ 
where, $\wh{g}_i$ and $\t{g}_i$ are the estimated channel gain and the unknown error term, respectively. For $i,j \in \mathcal{N}$, $\t{g}_i$ and $\t{g}_j$ are independent. The error $\t{g}_i$ is a circularly symmetric Gaussian random variable, i.e., $\t{g}_i \sim \mathcal{CN}(0,\epsilon_i^2),\: \forall i$. Upon transmitting the vector source signal $\mb{x}=[{x_1,x_2,\cdots,x_N}]^T$, the destination and the eavesdropper receive the following signals:
\begin{align*}
y_{d,i}=h_ix_i+z_{d,i}\; \text{ and } y_{e,i}=g_ix_i+z_{e,i},\; \forall i\in \mathcal{N}.
\end{align*}  
The noise variables $z_{d,i}$ and $z_{e,i}$ are i.i.d. across the $N$ parallel channels, the channel uses over time, and independent of the source signal. All noise variables are circularly symmetric Gaussian random variables with mean $0$ and variance $1$. Also, for practical reasons, we have a common power constraint over the parallel channels, i.e.,
$\sum_{i=1}^{N}\mathbb{E}[x_{i}^2] \le P$.
This assumption is quite practical when the transmitter has limited power supply; also, excessive power use can interfere with other transmitting nodes in radio range. 

For parallel independent Gaussian channels, secrecy capacity is attained when each source signal is distributed according to the Gaussian distribution, i.e., $x_i \sim \mathcal{CN}(0,P_i)$. Therefore we can write \cite{zang2010}:
%\vspace*{-1ex}

\small
\begin{align}\label{eq:secap}
\vspace*{-2ex}
\hspace*{-7pt}C_s\hspace*{-2pt}=\hspace*{-2pt}\underset{\bP \in \cP }{\max}\sum_{i=1}^N\left[\frac{1}{2}\log{(1+\hspace*{-2pt}|h_i|^2P_i)}\hspace*{-2pt}-\hspace*{-2pt}\frac{1}{2}\log{(1+\hspace*{-2pt}|g_i|^2P_i)}\right]^+
\vspace*{-3ex}
\end{align}

where, $\cP:=\{\bP: \sum\limits_{i=1}^NP_i \le P,\;P_i\ge 0, \forall i\}$,  $\bP=[P_1,P_2,\cdots,P_N]^T$ and $[x]^+=\max\{0,x\}$.

\vspace*{-2ex}
\normalsize
\section{Complete Secrecy Outage Analysis}
\label{sec:ana}
In Sec. \ref{sec:cso_inst}, we evaluate the instantaneous complete outage probability $\Pr(C_s=0|h_i,\wh{g}_i,\forall i)$\footnote{The precise expression is: $\Pr(C_s=0|H_i=h_i,\wh{G}_i=\wh{g}_i,\forall i)$. Here, $H_i$ and $G_i=\wh{G}_i+\t{g}_i$ are the random variables corresponding to the source-destination and the source-eavesdropper channels respectively. The source-destination channel gain $H_i=h_i$ is known perfectly and estimated source-eavesdropper channel gain is $\wh{G}_i=\wh{g}_i$ at each time epoch. Note that, $\t{g}_i$ is the uncertainty term and therefore, a random variable. For brevity, we have used the compression $\Pr(C_s=0|h_i,\wh{g}_i)$ instead.}---a consequence of imperfect information about eavesdropper's CSI. Sec. \ref{sec:cso_fad} computes the same for fading channels. 
\vspace*{-2ex}
\subsection{Complete Secrecy Outage for instantaneous channel gains, $\Pr(C_s=0|h_i,\wh{g}_i,\forall i)$}\label{sec:cso_inst}
Complete secrecy outage occurs when the receiver's absolute channel gain is less than the corresponding eavesdropper's absolute channel gain, for all channels, i.e., $|h_i| \le |g_i|,\;\forall i \in \mathcal{N}$. This scenario leads to \emph{zero} secrecy rate irrespective of the power allocated.\\
$\hspace*{5ex} \Pr(C_s=0|h_i,\wh{g}_i,\forall i)
 =\Pr(|h_i|<|g_i|,\forall i)\\\hspace*{25ex}=\prod\limits_{i=1}^{N}\Pr(|h_i|<|\wh{g}_i+\t{g}_i|)$.

The last equality is true because channels are independent.
We define random variables $X_{1i}:=\Re(\wh{g}_i+\t{g}_i) \sim \mathcal{N}(\Re(\wh{g}_i),\frac{1}{2}\epsilon_i^2)$ and $X_{2i}:=\Im(\wh{g}_i+\t{g}_i) \sim \mathcal{N}(\Im(\wh{g}_i),\frac{1}{2}\epsilon_i^2)$, where $\Re(.)$ and $\Im(.)$ are real and imaginary parts of a complex number, respectively. For each channel, the probability can be calculated in the following manner:\\
$ \Pr(|h_i|<|\wh{g}_i+\t{g}_i|) =\Pr(|h_i|^2<|\wh{g}_i+\t{g}_i|^2)\\ 
= \Pr\left( |h_i|^2 < \left( \frac{2X_{1i}^2}{\epsilon_i^2}+\frac{2X_{2i}^2}{\epsilon_i^2}\right) \frac{\epsilon_i^2}{2}\right) 
 =\Pr \left( |h_i|^2 < \chi_i^2\frac{\epsilon_i^2}{2}\right).$\\
Here, $\chi_{i}^2$ is a non-central chi-square random variable with degrees of freedom (d.o.f) $2$ and non-centrality parameter $\lambda_i^2=2\left( \frac{\Re(\wh{g}_i)}{\epsilon_i}\right) ^2+  2\left( \frac{\Im(\wh{g}_i)}{\epsilon_i}\right) ^2= \frac{2|\wh{g}_i|^2}{\epsilon_i^2}$. A non-central chi-square random variable can be approximated by a central chi-square random variable as follows \cite{coxCJS1987}:
\begin{align*}
\Pr\left( \chi_i^2< \eta_i \right) \approx \Pr\left( \chi_{i,0}^2< \frac{\eta_i}{(1+\lambda_i^2/2)}\right)=1-e^{-\frac{\eta_i}{2(1+\lambda_i^2/2)}}. 
\end{align*}
For small values of centrality parameter ($\lambda_i^2<0.4$), this approximation is quite accurate and for higher values it is conservative. 
The last equality is because a central chi-square random variable with d.o.f. 2, is distributed exponentially. Therefore, the final outage probability is
\begin{align*}
\prod\limits_{i=1}^{N}\Pr\left( \chi_i^2 >  \frac{2|h_i|^2}{\epsilon_i^2}\right)  \approx \prod\limits_{i=1}^{N}e^{-\frac{|h_i|^2}{\epsilon_i^2(1+\lambda_i^2/2)}}  =e^{-\sum\limits_{i=1}^{N}\frac{|h_i|^2}{|\wh{g}_i|^2+\epsilon_i^2}}.
\end{align*} 
 In Fig. \ref{fig:cmpr_outage}, we plot the outage probability calculated from simulation and analytical approximation with respect to $\epsilon_i^2$ for several values of $N$ and $|\wh{g}_i|$. Except for the initial part, where $\lambda_i^2 \nless 0.4$, the approximation is close to the simulation.
 
\begin{figure}[htbp]
	\centering
	\includegraphics[scale=0.45]{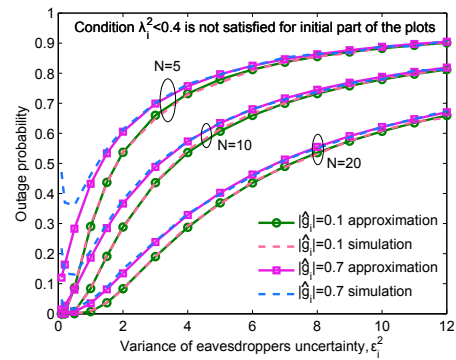}
	\caption{\small Plot of outage probability for numerical simulation and approximation with respect of $\epsilon_i^2$. Here, $|h_i|=0.5,\;\forall i$. 
	} 
	\label{fig:cmpr_outage}
\vspace*{-2ex}
\end{figure}
\vspace*{-4ex}
\subsection{Complete secrecy outage for fading channels $\Pr(C_s=0)$}\label{sec:cso_fad}
When $h_i$ and $\wh{g}_i$ are sampled from  circularly symmetric Gaussian distributions $\mathcal{CN}(0,\sigma_{m,i}^2)$ and $\mathcal{CN}(0,\sigma_{e,i}^2)$, respectively, then $|h_i|$ and $|g_i|$ have Rayleigh distributions with parameters $\sigma_{m,i}/\sqrt{2}$ and $(\sigma_{e,i}+\epsilon_i)/\sqrt{2}$. As $|h_i|^2$ and $|g_i|^2$ are exponentially distributed with parameters $1/\sigma_{m,i}^2$ and $1/(\sigma_{e,i}+\epsilon_i)^2$, respectively, following \cite{barrosISIT2006}:\\
$\displaystyle \Pr(C_s=0)=\prod\limits_{i=1}^{N}\left( \frac{1}{1+\rho_i}\right), \text{ where } \rho_i=\frac{\sigma_{m,i}^2}{(\sigma_{e,i}+\epsilon_i)^2}.
$
\vspace*{-2ex}
\section{Analysis for non-zero secrecy rate }
\label{sec:nonzerosec}
In Sec. \ref{sec:ana}, non-zero secrecy is not possible (irrespective of the power allocation). In contrast, in Sec. \ref{sec:rob_pow}, we formulate and provide the optimal solution for a robust power allocation problem to minimize the outage probability for a target secrecy rate, $R_s^{(0)}$, i.e., $\Pr(R_s<R_s^{(0)}|h_i,\wh{g}_i,\forall i)$. An expression for main channel outage for a target rate $R_s^{(0)}$ is also provided (Sec. \ref{sec:mainchnloutge}).
Only for this section, we use a low SNR approximation of Eq. \eqref{eq:secap}. This is valid for small values of $P$---typical in low power devices, such as, small sensor nodes deployed for surveillance. {Several commercial transceivers used in sensor nodes (e.g., ADF7020, ATA542X and CC1000 Series) have linear characteristics for significant portion of SNR \cite{sharmaTWC2010}.} 
\vspace*{-3ex}
\subsection{Robust optimal power allocation}
\label{sec:rob_pow}
Using $\ln(1+x)\approx x$, for $x \to 0$, we can approximate the secrecy rate in Eq. \eqref{eq:secap} as:
\vspace*{-1ex}
\begin{align}\label{opt:lowsnrscap}
R_s &= \frac{1}{2\ln (2)}\sum\limits_{i=1}^{N}\left[ |h_i|^2-|\wh{g}_i+\t{g}_i|^2\right] ^+P_i. 
\end{align}

For scenarios when non-zero secrecy is possible, we minimize the outage with respect to a target secrecy rate, $R_s^{(0)}$. As eavesdropper's CSI is imperfect, we can not directly maximize $R_s$; therefore, we need a robust power allocation approach to minimize the outage probability. The optimization problem can be written as:
\begin{align}\label{opt:outage}
\min\;\Pr(R_s<R_s^{(0)}|h_i,\wh{g}_i,\forall i), \text{ subject to: } \bP\in \mathcal{P}. 
\end{align} 
This optimization is needed only when the main channel, i.e., source to destination channel can sustain a rate  $R_s^{(0)}$. Otherwise, the secrecy rate is assured to be less than $R_s^{(0)}$ and the objective $\Pr(R_s <R_s^{(0)})=1$ everywhere in the \textit{feasible set}, $\bP\in \mathcal{P}$. For small SNR, the main channel can sustain a rate $R_s^{(0)}$ when
\begin{align}\label{eq:validity}
\sum\limits_{i=1}^{N}|h_i|^2P_i -2\ln(2)R_s^{(0)} >0.
\end{align} 
Here, the secrecy outage can be calculated from Eq. \eqref{opt:lowsnrscap} as:

\small
\begin{align}
& \Pr(R_s<R_s^{(0)}|h_i,\wh{g}_i,\forall i) \label{eq:outage_instn1} \\ 
& \le \Pr\left( \sum\limits_{i = 1}^{N}|h_i|^2P_i -2\ln(2)R_s^{(0)} < \sum\limits_{ i= 1}^{N}|\wh{g}_i+\t{g}_i|^2P_i\right) \label{eq:outage_instn2} \\
& = \Pr\left( \sum\limits_{i= 1}^{N}|h_i|^2P_i -2\ln(2)R_s^{(0)} < \sum\limits_{i= 1}^{N}\chi_{i}^2\frac{\epsilon_iP_i}{2}\right). \nonumber
\end{align}
\normalsize
As discussed in Sec. \ref{sec:cso_inst}, $\chi_{i}^2$ is a non-central chi-square random variable (d.o.f. 2)\footnote{We emphasize that non-central to central chi-square approximation is \emph{not} used/required in this section.}. Therefore, the mean of $\chi_{i}^2\frac{\epsilon_iP_i}{2}$ is
\begin{align*}
\mathbb{E}[\chi_{i}^2\frac{\epsilon_iP_i}{2}] =\hspace*{-2pt}\frac{\epsilon_iP_i}{2}(2+\lambda_i^2) =\hspace*{-2pt} \frac{\epsilon_iP_i}{2}(2+\frac{2|\wh{g}_i|^2}{\epsilon_i^2}) = \hspace*{-2pt}(\epsilon_i^2+|\wh{g}_i|^2)P_i.
\end{align*}
Using the Markov inequality, we can bound the outage probability as follows:
\vspace{-0.35cm}

\small
\begin{align}\label{eq:markov}
& \Pr\Bigg( \sum\limits_{i=1}^{N}\chi_{i}^2\frac{\epsilon_iP_i}{2} >\underbrace{\sum_{i=1}^{N}|h_i|^2P_i -2\ln(2)R_s^{(0)}}_{\Delta}  \Bigg) \\ 
\le & \frac{\mathbb{E}\left[ \sum\limits_{i=1}^{N}\chi_{i}^2\frac{\epsilon_iP_i}{2}\right]}{\sum\limits_{i=1}^{N}|h_i|^2P_i -2\ln(2)R_s^{(0)}} \nonumber 
= \frac{\sum\limits_{i=1}^{N}(\epsilon_i^2+|\wh{g}_i|^2)P_i}{\sum\limits_{i=1}^{N}|h_i|^2P_i -2\ln(2)R_s^{(0)}}.
\end{align}
\normalsize
The validity of Markov bound requires $\Delta\ge 0$ in Eq. \eqref{eq:markov}; this is ensured by Eq. \eqref{eq:validity}. The Markov bound is known to be loose; however, we believe that the proposed power allocation scheme has value as substantiated by the numerical results in Sec. \ref{sec:Numerical}, which show improvement over several conventional schemes. In addition, this approach leads to an easy-to-compute solution (Proposition \ref{prpn:optpow}) on resource-constrained sensor nodes. 

To minimize the outage, we propose to minimize this upper bound. This leads to the following equivalent linear fractional program (approximation to Problem \eqref{opt:outage}):
\begin{align}\label{opt:linfrc}
\min \;\; \frac{\mb{c}^T\mb{P}}{\mb{d}^T\mb{P}-2\ln(2)R_s^{(0)}},
\text{ subject to:}\;\; \bP\in \mathcal{P}.
\end{align}
Here, $c_i=\epsilon_i^2+|\wh{g}_i|^2$, $d_i=|h_i|^2$, and $\mb{c}$, $\mb{d}$ are vectors of these elements. The denominator $\mb{d}^T\mb{P}-2\ln(2)R_s^{(0)}\ne0$ from the earlier discussion. This linear fractional program can be solved numerically by reformulating it as a linear program using the \textbf{Charnes-Cooper} transformation. However, due to the simplex constraint and only a few variables (the number of parallel channels, $N$), we propose the following easy to compute solution\footnote{Unlike our case, when the number of corner points is large, simplex or interior-point methods are efficient.}.
\begin{prpn}\label{prpn:optpow}
	The optimal solution to Problem \eqref{opt:linfrc} lies in one of the corners of the set $\overline{\mathcal{P}}:=\{\bP: \sum\limits_{i=1}^NP_i = P,\;P_i\ge 0, \forall i\}$, i.e., $P_i=P$ for some $i\in \mathcal{N}$ and $P_j=0,\; \forall j \ne i$.
\end{prpn}  
\begin{proof} We provide an outline of the proof here. One can verify that, at the optimum, the objective function of Problem \eqref{opt:lowsnrscap} will consume the total budget $P$ (when at least one of the coefficients of $P_i$ is non-zero). Therefore, we can use equality in the sum constraint. From \cite{swarupOR1965}, a linear fractional program attains its optimum at the basic feasible solution of the constraint set\footnote{This is easy to see because the gradient of the objective function in \eqref{opt:linfrc} is non-zero in the constraint set, leading to behavior similar to linear programs.}. The corners of the constraint hyperplane---only $N$ in number---are the basic feasible solutions and the one that minimizes the objective function is the optimum.
\end{proof}

Note that, unlike a general linear (fractional) program where calculating the corner points is computationally costly, in our case, the corners are known and fixed over the parameter set (specified in Proposition \ref{prpn:optpow}). Thus, the optimum is computed by enumerating the objective value at each corner point and selecting the best channel---computationally efficient even for a sensor node.

\subsection{Main channel outage for fading scenarios (for small SNR)}\label{sec:mainchnloutge}
For completeness, we evaluate the probability of the event when the optimization problem \eqref{opt:linfrc} need not be solved, as the main channel capacity itself is less than the target rate. The fading coefficients $|h_i| \sim Rayleigh(\sigma_{m,i}/\sqrt{2}),\forall i$. The outage occurs when the strongest of the $N$ parallel channels cannot sustain the target rate, i.e.,

\small
\begin{align*}
&\Pr\left(\max_i\;\left\lbrace |h_i|^2P\right\rbrace - 2\ln(2)R_s^{(0)} < 0 \right)\\
& = \Pr\left(\bigcap\limits_{i=1}^N \left\lbrace |h_i|^2P < 2\ln(2)R_s^{(0)}\right\rbrace  \right) \\
& \overset{(u)}{=} \prod\limits_{i=1}^{N}\left( |h_i|^2P < 2\ln(2)R_s^{(0)} \right)
\overset{(v)}{=}\prod\limits_{i=1}^{N}\left( 1-e^{-\left( \frac{2\ln(2)R_s^{(0)}}{P\sigma_{m,i}^2}\right)} \right) .
\end{align*}

\normalsize
$(u)$ and $(v)$ are true because $h_i,\;\forall i$ are independent and $|h_i|^2$ follows the exponential distribution.

\begin{figure*}[!t]
\vspace*{-1ex}
	\centering	
	\subfloat[Outage vs. $R_s^{(0)}$ \label{fig:R0_var}]{
		\includegraphics[width=64mm]{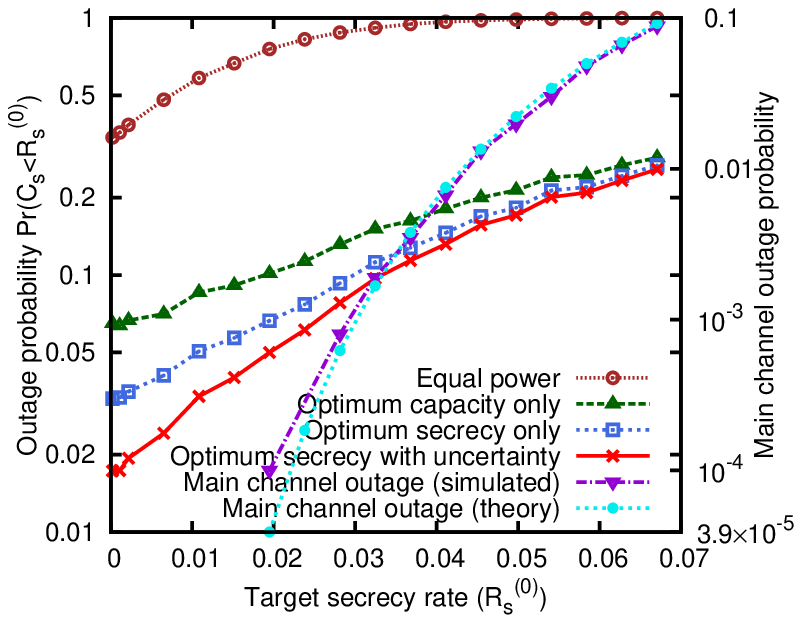} }
	\hspace*{-0.7cm}\hfill\hspace*{-0.7cm}
	\subfloat[Outage vs. eavesdropper's channel uncertainty \label{fig:eps_var}]{
		\includegraphics[width=64mm]{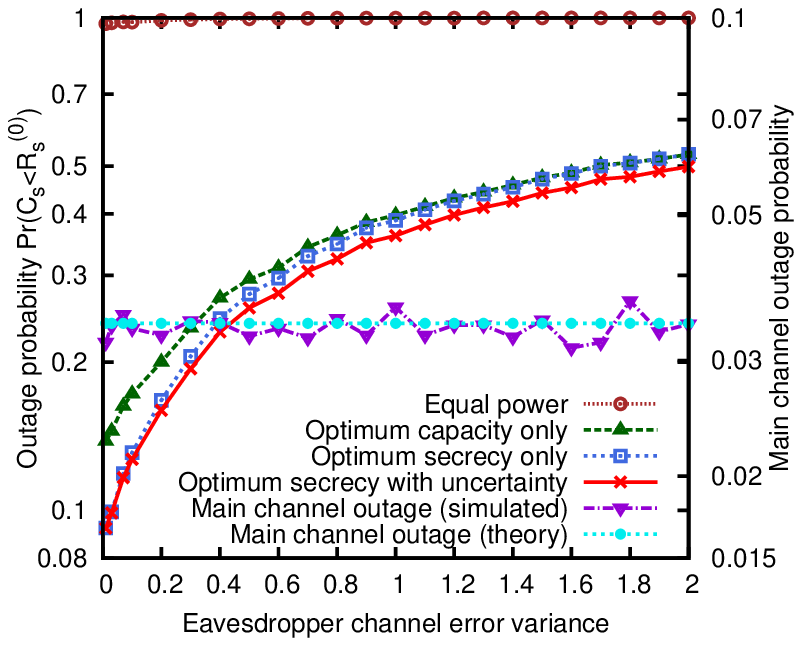} }	
	\hspace*{-0.7cm}\hfill\hspace*{-0.7cm}
	\subfloat[Outage vs. $N$ \label{fig:N_var}]{
		\includegraphics[width=64mm]{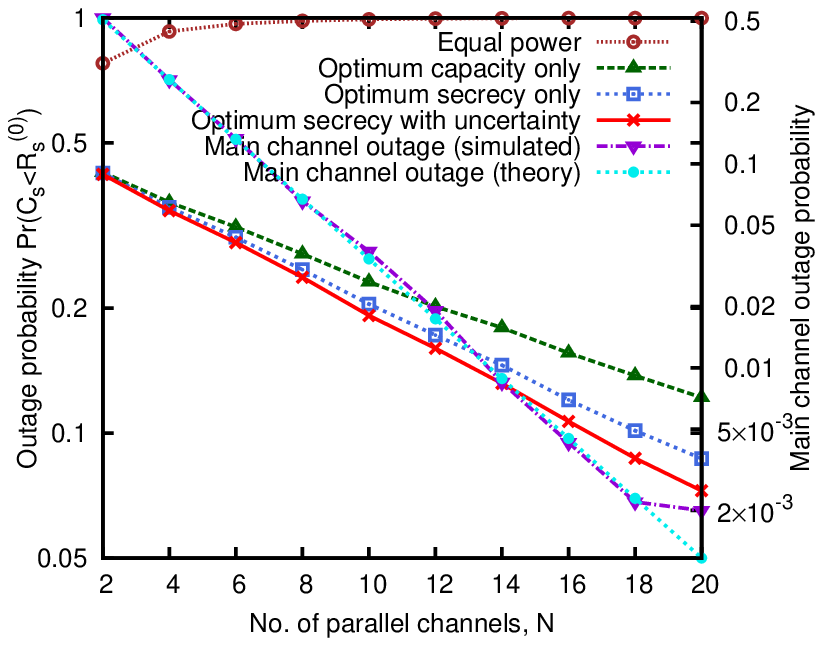} }
	\vspace*{-2ex}
	\caption{\small{Comparison of secrecy outage with respect to different system parameters. Default parameter values (except the one varied): $P=0.1$, $\sigma_{m,i}^2=0.6$, $\sigma_{e,i}^2=0.3$, $R_s^{(0)}=0.625P\times\sigma_{m,i}^2/2\ln(2)$, $\epsilon_i^2=0.3$ for $i \in \{1,2,\cdots,\lceil N/2\rceil\}$ and $\epsilon_i^2=0.09$ for $i \in \{\lceil N/2 \rceil +1,\cdots,N\}$ and $N=10$.  }}
	\label{fig:outage_vs_param}
\vspace*{-4ex}
\end{figure*}

\section{Numerical Results}
\label{sec:Numerical}
In this section, we compare several power allocation schemes for parallel independent Gaussian channels and show the effectiveness of our proposed scheme that considers eavesdropper's channel uncertainty (optimization problem (8)). We consider three conventional power allocation strategies for comparison: (a) \textit{Equal Power}---allocates equal power on every channel, $P_i=P/N,\;\forall i\in \mathcal{N}$. (b) \textit{Optimum Capacity} power allocation---maximizes the Shannon capacity in the main channel, i.e., $\bP_c^* = \arg \max \sum_{i=1}^{N}|h_i|^2P_i$ (assuming small SNR), where $\bP_c^* \in \mathcal{P}$. It allocates the power budget $P$ to the strongest main channel. (c) \textit{Optimum secrecy} power allocation---maximizes secrecy rate for the estimated eavesdropper channel gain without considering uncertainty. $\bP_s^* = \arg \max \sum_{i=1}^{N}[|h_i|^2-|\wh{g}_i|^2]^+P_i$, where $\bP_s^* \in \mathcal{P}$. It allocates $P$ to the channel that has largest value of $|h_i|^2-|\wh{g}_i|^2$. 

For simulations, we have generated the main channel $h_i$, and the estimated eavesdropper channel $\wh{g}_i$, from $\mathcal{CN}(0,\sigma_{m,i}^2)$ and $\mathcal{CN}(0,\sigma_{e,i}^2)$, respectively with $\sigma_{m,i}^2=0.6$ and $\sigma_{e,i}^2=0.3$ for all parallel channels. The power budget considered is $P=0.1$. We use the default parameters for the target secrecy rate $R_s^{(0)}=0.625P\times\sigma_{m,i}^2/2\ln(2)$ and the number of parallel channels $N=10$. Uncertainty in the $i$th eavesdropper's channel is generated from $\mathcal{CN}(0,\epsilon_i^2)$ with $\epsilon_i^2=0.3$ for $i \in \{1,2,\cdots,\lceil N/2\rceil\}$ and $\epsilon_i^2=0.09$ for $i \in \{\lceil N/2 \rceil +1,\cdots,N\}$.

Fig. \ref{fig:outage_vs_param} plots the variation of the objective function, secrecy outage $\Pr(R_s <R_s^{(0)})$, with respect to variation of the system parameters (left y-axis); and the main channel outage for the same target rate $R_s^{(0)}$ (right y-axis). In all the three figures, ``Optimum secrecy with uncertainty" identifies the proposed scheme. The plot  $\Pr(R_s <R_s^{(0)})$ is calculated, for example, for the proposed scheme as follows. 

We generate $h_i,\;\wh{g}_i,\;i \in \mathcal{N}$. Given this channel, we calculate the power allocation of the proposed scheme $\bP_p^*$ using the optimization problem \eqref{opt:linfrc}. Note that, we only need the variance of the uncertainty of the eavesdropper's channel, $\epsilon_i^2,\;\forall i$ and not the exact value of $g_i$ to calculate $\bP_p^*$. Now, the conditional outage probability given $h_i$ and $\wh{g}_i,\;\forall i$, i.e., $\Pr(R_s <R_s^{(0)}|h_i,\wh{g}_i,\forall i)$ is calculated by Monte-Carlo averaging over $10^4$ values of $\t{g}_i,\forall i$ using \eqref{eq:outage_instn2} for $\bP_p^*$. We obtain the final objective function, $\Pr(R_s <R_s^{(0)})$, as follows: $10^4$ instances of $h_i$ and $\wh{g}_i$ are generated, \eqref{eq:outage_instn2} is averaged over the instances for which the main channel is not in outage, i.e., $\max\{|h_i|^2,\forall i\}P> 2\ln(2)R_s^{(0)}$ . The same is carried out for the three conventional power allocation strategies. 

Fig. \ref{fig:R0_var} demonstrates the effectiveness of the proposed scheme when the target secrecy rate $R_s^{(0)}$ is varied, specially for small values of $R_s^{(0)}$. The effect of variation of eavesdropper's channel uncertainty is shown in Fig \ref{fig:eps_var}. As expected, when uncertainty is small, performance of the ``Optimum secrecy" power allocation is comparable to the proposed scheme. However, as the uncertainty increases, the proposed scheme outperforms the others. Finally, we vary the number of channels, $N$ in Fig. \ref{fig:N_var}. The proposed scheme exploits the uncertainty parameter $\epsilon_i^2$, in addition to the channel gains, and therefore, outperforms ``Optimum secrecy" (and others). The improvements are prominent for higher values of $N$.
%\vspace*{-3ex}
\section{Conclusion}
\label{sec:Concl}
We study the effect of a single eavesdropper's channel uncertainty in a parallel independent Gaussian channel communications system. We evaluate complete secrecy outage probability for instantaneous and fading channels. Further, we propose a robust power allocation scheme to minimize secrecy outage probability at a target rate in the low SNR regime. Our techniques involve (1) approximating non-central chi-square random variables by corresponding central ones to evaluate instantaneous outage; and (2) bounding outage probabilities for robust power allocation, which leads to a linear fractional program. Exploiting the structure of the problem, we propose an easy-to-compute solution. Numerical results show the superiority of the proposed scheme compared to the conventional schemes that do not consider uncertainty. 

\footnotesize 
\bibliographystyle{IEEEtran}
\bibliography{testbib} 

\end{document}